\documentclass[11pt, a4paper,onecolumn]{belarticle2}
\usepackage{revsymb, amsmath, amsfonts, amssymb, enumerate, fullpage, amsthm, graphicx, braket, relsize, bbm, mathrsfs}
\usepackage{graphicx,color}
\usepackage{paralist}

\newtheorem{theo}{Theorem}
\newtheorem{thm}[theo]{Theorem}

\newcommand{\cO}{\mathcal{O}}

\newcommand{\tr}[2]{\mathrm{tr}_{#2} \left\{ #1 \right\}}
\newcommand{\trace}[1]{\mathrm{tr}\left\{#1 \right\}}
\newcommand{\proj}[1]{\ket{#1}\!\bra{#1}}
\newcommand{\half}{$\frac{1}{2}$ }
\newcommand{\av}[1]{\langle #1 \rangle}

\begin{document}

\title{Quantum Reference Frames \protect\\ and Their Applications to Thermodynamics}

\author{Sandu Popescu}
\affiliation{H. H. Wills Physics Laboratory, University of Bristol, Tyndall Avenue, Bristol, BS8 1TL, United Kingdom}
\email{s.popescu@bristol.ac.uk}

\author{Ana Bel\'{e}n Sainz}
\affiliation{Perimeter Institute for Theoretical Physics, Waterloo, Ontario, Canada, N2L 2Y5}
\email{sainz.ab@gmail.com}

\author{Anthony J. Short}
\affiliation{H. H. Wills Physics Laboratory, University of Bristol, Tyndall Avenue, Bristol, BS8 1TL, United Kingdom}
\email{tony.short@bristol.ac.uk}

\author{Andreas Winter}
\affiliation{ICREA---Instituci\'o Catalana de Recerca i Estudis Avan\c{c}ats, Pg.~Lluis Companys 23, 08010 Barcelona, Spain}
\affiliation{F\'{\i}sica Te\`{o}rica: Informaci\'{o} i Fen\`{o}mens Qu\`{a}ntics, Departament de F\'{\i}sica,\\[-1mm] Universitat Aut\`{o}noma de Barcelona, 08193 Bellaterra (Barcelona), Spain}
\email{andreas.winter@uab.cat \\[1.33mm] AW spoke on quantum thermodynamics at the Royal Society discussion meeting \emph{Foundations of quantum mechanics and their impact on contemporary society}, 11-12 December 2017, and reported, among other things, on the content of \cite{NCCQ,jonathan}; the present paper addresses some of the issues highlighted there.}

\date{\today}

\maketitle

\begin{abstract}
We construct a quantum reference frame, which can be used to approximately
implement  arbitrary unitary transformations on a system in the 
presence of any number of extensive conserved quantities, by absorbing 
any back action provided by the conservation laws. Thus, the reference
frame at the same time acts as a battery for the conserved quantities.

{Our construction features a  physically intuitive, clear and implementation-friendly
realisation. Indeed, the reference system is composed of the same types of subsystems as the original system
and  is finite for any desired accuracy.
In addition, the interaction with the reference frame can be broken down into
two-body terms coupling the system to one of the reference frame subsystems at a time.}
We apply this construction to quantum thermodynamic setups with multiple,
possibly non-commuting conserved quantities, which allows for the 
definition of explicit batteries in such cases.
\end{abstract}

\section{Introduction}
\label{sec:intro}
When considering what transformations of a quantum system are theoretically possible, we usually imagine that we can implement any unitary transformation on the system. However, this may not be consistent with conservation laws:  Wigner \cite{Wig52} and Araki and Yanase \cite{Ara60} showed, in what is known as the WAY theorem, that an observer cannot measure any observable that does not commute with an additive conserved quantity, as such a measurement would necessarily change this conserved quantity. From an information-theoretic perspective, WAY is a consequence of the no-programming theorem for projective measurements \cite{Marvian}. The situation was revisited, and clarified, by Aharonov and Susskind~\cite{AharonovSusskind1,AharonovSusskind2}, followed by Aharonov and Kaufherr~\cite{AharonovKaufherr} who pointed out its deep connection with the issue of quantum frames of reference: The conservation laws always refer to a closed system, such as a freely floating spaceship, and there is a need to consider carefully the difference between observables defined in relation with frames of reference defined internally (i.e.~inside the rocket) or externally. This enabled them to also clarify the issue of super-selection rules~\cite{wig}, another apparent constraint on quantum mechanics. Recently, there has been a renewed very strong interest in understanding quantum frames of reference as well as their connection with conservation laws; see e.g. \cite{RobRTA,dialogue}.

All the above-mentioned works address the fundamental question of which operations are physically possible on a system under conservation laws. 
{The present paper focuses on one of the central issues in this area - how to implement a unitary transformation  on the system with respect to an external reference frame. Previous approaches to this problem have either considered particular cases, or relied upon general group theoretic constructions for an internal reference frame. This latter approach requires the existence of a very abstract and complex ancillary system to act as the reference frame, which interacts in a complex way with the system. Our approach, in contrast, gives a general construction for an internal reference frame which is also easy to implement and understand physically.}
 In particular, we use  
a single fixed reference frame to perform an arbitrary unitary transformation 
with high precision on an unknown state of a system;
the same reference frame will work for any extensive conserved quantities, 
and does not depend otherwise on the symmetries to which they are 
linked. Indeed, we could impose an arbitrary set of extensive conservation laws 
which are not even linked to symmetries of nature, and the construction would still apply; this is for example the case 
in the applications of the theory to thermodynamics, to be discussed later. 
Furthermore, the reference frame (i) is composed of the same basic subsystems 
as the system on which it acts (and thus does not require the existence of 
a new type of system with special properties), (ii) is finite for any desired 
precision, and (iii) the required transformations can be implemented by a series of  
interactions between two subsystems at a time, each step being a very short
time evolution of a simple Hamiltonian between the system and one of its
replicas in the reference frame (Sec.~\ref{se:frame}). 
We also show that if the system itself
is composed of smaller parts, the interaction with the reference frame 
can be implemented by a quantum circuit in which each gate respects the 
conservation laws (Subsec.~\ref{sec:composite}).

Our motivation for the present work stems from recent research on quantum 
thermodynamics, and in particular by considering the thermodynamics of systems 
with multiple conserved quantities \cite{NCCQ,jonathan}. Given access to a 
large generalised thermal state (composed of many identical subsystems 
with a small amount of assumed structure), and the ability to perform any 
unitary, it is possible to extract an arbitrary amount of any given 
conserved quantity by inputting an appropriate amount of the other 
conserved quantities \cite{NCCQ}. A generalised second law restricts 
the rate at which different conserved quantities can be interconverted. 
In this picture, any change in the conserved quantities of the system  
due to the action of the unitary is interpreted as work (implicitly 
stored in some battery).  If the conserved quantities all commute, 
then the formalism can be extended to include explicit batteries, 
with the system+battery combination obeying strict conservation laws. 
In \cite{NCCQ} it was left as an open problem how to achieve a similar 
result for non-commuting conserved quantities. A similar problem was 
addressed in \cite{jonathan}, where a reference frame was used as 
an explicit battery to implement such transformations. However, that 
work was based on interactions with a particular infinite dimensional 
choice of ``perfect'' reference system, which may not be available. 
{Here we use our alternative reference frame construction to present 
an explicit battery which is simply comprised of many copies of the 
subsystems on which it acts, in different initial states (Sec.~\ref{se:thermo}), 
and which can operate with arbitrary effectiveness while being of finite size. }

\section{Background}
\label{sec:back}

Consider a closed quantum system, such as the spaceship mentioned before, subjected to conservation laws. Consider a particle in that closed system, say a spin-\half particle. Suppose we want to implement a given unitary evolution, \emph{defined in terms of an external frame of reference}, say a rotation by an angle $\theta$ around the $z$-direction. Could this be achieved by apparatuses situated \emph{inside} the rocket?  The reason why this question received so much interest is that, apart from the problem itself and what it tells us about conservation laws and frames of reference, when cast in different set-ups, it has major implications for many other domains. For example, Kitaev \emph{et al.}~\cite{Kitaev} considered it in the context of quantum cryptography protocols; more recently, it has emerged as a crucial issue in quantum thermodynamics, cf.~\cite{LKJR}.

The point is that conservation laws impose constraints to what can be done: The desired unitary, { $U=\exp(-i\sigma_z\theta)$ } in the above example, does not commute with the total angular momentum along any other axis apart from $z$, say with $J_x^{\text{total}}$, hence it is impossible. Therefore, apparently there is no way to achieve our desired goal. It turns out, however, that there is a way: Following \cite{AharonovKaufherr}, the reason why a rotation along the external $z$ axis is impossible in general is that actions {confined} inside the closed system do not have access to the external frame of reference, so they have no knowledge of what the $z$ axis actually is. All that internal apparatuses can do is to rotate the spin around some internally defined axis, say the direction from the floor to the ceiling of the rocket: Such a rotation, relative to the internal degrees of freedom -- the other parts of the rocket -- commutes with all components of total angular momentum so does not contradict the angular momentum conservation law. If however, in advance of the request for the rotation, the external observers align the internal direction -- the ``internal frame of reference'' -- with the external one, say the `floor to the ceiling' direction to the external $z$-direction, and if rotating the spin does not produce too much of a back-reaction to the rocket to misalign it, the rotation along the internal direction coincides with a rotation around the $z$-axis.

The central question is thus how to prepare and align the internal frame, which, for more general unitaries, is not so obvious as in the case of spin rotations, and how to ensure it is robust enough, so that it stays aligned under the kick-back. 
{Building up from the WAY theorem, new works have further explored complex mathematical constructions to devise internal frames. For instance, a general version of this problem was also studied as the resource theory of asymmetry \cite{RobRTA, Ahmadi, Marvian13}. } Other approaches have further focused on the role of the particular symmetry group $G$ which, according to Noether's theorem, comes with the conservation laws. As a consequence, the internal reference frame that one can define following this approach is heavily group dependent. This fact is clearly illustrated by the construction of the ``model'' reference frame $L^2(G)$: the square integrable functions on the group $G$ (under the invariant Haar measure) equipped with the left regular representation \cite{RobRTA}. Hence, while this approach is mathematically very elegant, both in its definition and its interaction with the system, {the construction is very complex, and so is the interaction required to implement the desired unitary on the system. Hence one might wonder what sort of  complex physical system this frame actually represents, and how to physically implement the necessary interaction. }

{In the next {section} we show how one can construct instead a reference frame that is both simple to implement and physically intuitive, { and which will work for any extensive conserved quantities}}. 

\section{A general quantum reference frame}
\label{se:frame}
In this section, we construct a general quantum reference frame which will allow us to perform an arbitrary unitary transformation $U_S$ with high precision on an unknown system state $\rho_S$ despite arbitrary extensive conserved quantities. 

The reference frame, $R$, is composed of a large number of copies of the system, and is  prepared in a particular state $\rho_R$  which is independent of $U_S$ or $\rho_S$. The size of the reference frame will depend on the precision to which we will attempt to simulate the unitary $U_S$, and {will be made more precise in appendix \ref{ap:On}}. 

We define an extensive conserved quantity in this context as one which can be written as  
$ A^{\text{TOT}}=A_{S} + \sum_{r \in \text{R}} A_r$. In other words, as a sum of the same  quantity $A$ for the system and for each subsystem $r$ in the reference frame. 

We show below that any $U_S$ can be implemented  on the system with high precision by an appropriately chosen joint unitary $V$ on the combined system and reference frame which commutes with all conserved quantities. 

\begin{thm} 
Let $S$ be any finite dimensional quantum system. Then for every $\epsilon > 0$, there exists a reference frame (composed of a large number of copies {of the system $S$, prepared in suitable states}) with a fixed state $\rho_R$,  such that for every unitary $U_S$ on the system  there exists a joint unitary $V$  on the system and reference frame such that
\begin{itemize}
\item $V$ conserves all extensive conserved quantities: $[V, A^{\text{TOT}}] = 0$.
\item $V$ effectively implements $U_S$ on the system with $\epsilon$ precision, for any initial system state: 
\begin{equation}
 \left\| \tr{V \, \rho_S \otimes \rho_R\, V^\dagger }{R} - U_S \rho_S U_S^\dagger \right\|_1 \leq \epsilon \quad \forall \rho_S,
\end{equation}
where $\rho_S$ and $\rho_R$ are density operators (i.e. $\rho \geq 0$, $\mathrm{tr}\,\rho=1$),
and $\| \cdot \|_1$ is the trace norm (which characterises how well two states can be distinguished). 
\end{itemize}
\end{thm}

\begin{proof}
The proof consists of two parts. First, that infinitesimal rotations can be implemented, and then that from them a general unitary can be implemented without going above the precision threshold.

For the first part of the proof, we fix a large $N \gg 1$ and define the unitary $V_{\alpha}$ acting on the system and one copy of the system from the reference frame (denoted by $R$) as follows:
\begin{align} \label{eq:Vtheta} 
V_{\alpha} =\exp \left( -i \frac{\alpha}{N} \text{SWAP} \right) \,,
\end{align}
where $\text{SWAP} = \sum_{i,j} \ket{ij}\!\bra{j i}$. Since the conserved quantities are additive, and all subsystems are identical, any function of the SWAP operation will conserve them. Therefore, $[V_{\alpha} ,  A^{\text{TOT}}] = 0$ for all $\alpha$. 

Let $\sigma$ be an arbitrary but fixed initial state for the reference frame particle. The effective action on the system when applying $V_{\alpha}$ on $\rho_S \otimes \sigma$ is:
\begin{align} \label{eq:effectiveU} 
\rho_S^f &= \tr{ V_{\alpha} \, \rho_S \otimes \sigma \, V_{\alpha}^\dagger }{R} \nonumber \\
&= \tr{ \rho_S \otimes \sigma -i \tfrac{\alpha}{N} \, \text{SWAP} \, \rho_S \otimes \sigma + i \tfrac{\alpha}{N} \, \rho_S \otimes \sigma \, \text{SWAP} }{R} + \cO\left( \frac{1}{N^2} \right) \nonumber \\
&= \rho_S -i \tfrac{\alpha}{N} \, [\sigma, \rho_S] + \cO\left( \frac{1}{N^2} \right)\,,
\end{align}
where a Taylor expansion for small values of $\tfrac{\alpha}{N}$ was used.\footnote{We will say that a trace class operator $A$ is  $\cO\left( \frac{1}{N^k} \right)$ if $\|A\|_1 \leq c/N^k$ for all sufficiently large $N$, where $c$ is a constant (which here depends on $\alpha$). We show this in detail in Appendix \ref{ap:On}. In other contexts, we will use the same notation to refer to a number $a$ satisfying $|a| \leq c/N^k$ for all sufficiently large $N$.}
The last equality follows from $\tr{\text{SWAP} \, \rho_S \otimes \sigma}{B} = \sigma \, \rho_S$ and $\tr{\rho_S \otimes \sigma \, \text{SWAP}}{B} = \rho_S \, \sigma$ (see Appendix  \ref{ap:swap}).

Defining $U_{\sigma}(\alpha) = \exp\{ -i\frac{\alpha}{N} \sigma\}$ we thus find that  
\begin{equation*}
\rho_S^f =  U_{\sigma}(\alpha) \, \rho_S  \, U_{\sigma}^{\dagger}(\alpha) + \cO\left( \frac{1}{N^2} \right)\,, 
\end{equation*}
i.e.
\begin{align}\label{eq:epsfor1}
  \left\|\tr{V_{\alpha} \, \rho_S \otimes \sigma\, V_{\alpha}^\dagger }{B} - U_{\sigma}(\alpha) \, \rho_S  \, U_{\sigma}^{\dagger}(\alpha) \right\|_1 <  \cO\left( \frac{1}{N^2}\right).
\end{align}


The above procedure allows us to generate small rotations of a qubit about a particular state $\sigma$. We now show how to generalise this to implement arbitrary small rotations. 

Let $\mathcal{H}_S$ be the Hilbert space of the system,  and  $\{\sigma_k\}_{k=1}^D$ be a set of density operators describing  states of the system such that  $\mathcal{B} = \{\sigma_k\} \cup \mathbbm{1}$ is an arbitrary but fixed operator basis of $\mathcal{H}_S$\footnote{Note that this implies that $D=d^2-1$ where $d$ is the Hilbert space dimension.}.
Now consider {the following state of $R$},  composed of one copy of each of the states $\sigma_k$:
\begin{equation} 
  \sigma_R= \sigma_1  \otimes \sigma_2 \otimes \ldots \sigma_D.
\end{equation}

Suppose that we ultimately wish to implement an arbitary unitary $U_S = \exp\left( - i H \right)$ on the system. As global phase factors in $U_S$ do not affect the evolution of the state, we can assume without loss of generality that $H= \sum_{k=1}^D \alpha_k \sigma_k$. Additionally, as all values of $\exp(i \theta)$ can be obtained for $ \theta \in (-\pi, \pi]$, we can choose $H$ such that its eigenvalues lie within this range. As the system is finite dimensional, for any fixed operator basis $\mathcal{B}$, this also implies that the parameters  $\alpha_k$ are also bounded by a constant (see appendix \ref{ap:alpha_bound}). 

Let us now consider how to implement the small general rotation $U_H = \exp\left( - i \frac{H}{N} \right)$. Denote by $V_{\alpha_k}^{(k)}$ the unitary $V_{\alpha_k}$ applied on the system and the particle from the reference frame that is in state $\sigma_k$. Now consider the action of the sequence of unitaries $V_{\text{seq}} =V_{\alpha_D}^{(D)} \ldots V_{\alpha_2}^{(2)} V_{\alpha_1}^{(1)} $ on the system $\rho_S$ and reference frame $\sigma_R$. By keeping the first order terms in $\tfrac{1}{N}$, similarly to the calculation in the infinitesimal-rotation case, we obtain:
\begin{align}
\rho_S^f &= \tr{V_{\text{seq}}  \rho_S \otimes \sigma_R \, V_{\text{seq}}^\dagger}{R_1 \ldots R_D} \nonumber \\
&= \tr{ \rho_S \otimes \sigma_R -i \sum_{k=1}^D \tfrac{\alpha_k}{N} \, \text{SWAP}_{S,R_k} \, (\rho_S \otimes \sigma_R) + i \sum_{k=1}^D \tfrac{\alpha_k}{N} \, (\rho_S \otimes \sigma_R) \, \text{SWAP}_{S,R_k} }{R_1 \ldots R_D} + \cO\left( \frac{1}{N^2} \right) \nonumber \\
&= \rho_S -i \tfrac{1}{N} \, \left[\sum_{k=1}^D \alpha_k \sigma_k, \rho_S\right] + \cO\left( \frac{1}{N^2} \right) \nonumber\\
&= \rho_S -i \tfrac{1}{N} \, [H, \rho_S] + \cO\left( \frac{1}{N^2} \right) \nonumber \\
& = U_H \rho_S U^{\dagger}_H  + \cO\left( \frac{1}{N^2} \right)\,. \label{eq:nueqvoe}
\end{align}

We are now in position to prove the general claim of the Theorem, where 
$U_S$ is an arbitrary unitary on $S$ (which need not be close to the identity). 

Consider a reference frame  composed of $N $ copies of $\sigma_R$: 
\begin{equation} 
  \rho_R=\sigma_R^{\otimes N} =  (\sigma_1  \otimes \sigma_2 \otimes \ldots \sigma_k)^{\otimes N}.
\end{equation}
The required value of $N$ will depend on the precision to which we will 
attempt to simulate the unitary $U_S$.

To implement $U_S$, we apply the unitary $V_{\text{seq}}$ on the system plus $D$ reference particles $N$ times, each time using a different copy of $\sigma_R$  in the reference frame. Each unitary in the sequence will  act on the system as the unitary $U_H$ up to an error $\cO\left( \frac{1}{N^2} \right)$. The upper bound on the total error after applying the sequence of $V_{\text{seq}}$ will be $N \cO\left( \frac{1}{N^2} \right) \sim \cO\left( \frac{1}{N} \right)$ (see appendix \ref{concatenation}).  Hence 
\begin{equation} 
  \left\| \tr{V \, \rho_S \otimes \rho_R\, V^\dagger }{R} - U_S \rho_S U_S^\dagger \right\|_1 \leq  \cO \left( \frac{1}{N} \right) 
\end{equation} 

Given any $\epsilon>0$ and a unitary $U_S$, it is always possible to find a sufficiently large $N$ such that $\cO \left( \frac{1}{N} \right)  \leq \epsilon$. Hence a sufficiently large reference frame  will allow us to take the system to a final state $U_S \rho_S U_S^\dagger$ within an error upper bounded  by $\epsilon$ as required. This is achieved by applying a sequence of unitary transformations $V_{\alpha}$  on the system plus reference frame particles, each of which commutes with (and hence conserves) all extensive quantities $A^{\text{TOT}}$. 
\end{proof}

\subsection{Composite systems} 
\label{sec:composite} 
Although the previous result will work for any system, including composite systems, the construction will in general require  joint transformations involving all of the particles in the system. Particularly in the context of thermodynamics, it is common to consider systems composed of many particles, where this may pose practical difficulties. Here we show that our approach can be adapted easily to allow for a `circuit'-model of transformations, where any unitary on the system can be implemented by a series of transformations involving a small number of particles from the system and reference frame. 

{ If we have a composite system made up of some finite number of different types of subsystems, then our reference frame is simply  the tensor product of the reference frames for each type of subsystem.  Using our previous results, we can approximately implement any single-subsystem unitary. To implement any two-subsystem unitary between subsystems $A$ and $B$, we {(i)} consider those two subsystems as a composite system, and {(ii)} use the same procedure as before on this composite system and a similar composite system in the reference frame { (hence involving only 4 primitive subsystems).} For example, to implement $\exp \left( - i \alpha \frac{\sigma_1^{(A)} \otimes \sigma_1^{(B)}}{N} \right)$ we would apply 
$\exp \left( - i \frac{\alpha}{N} \, \text{SWAP}_{S_A, R_A} \, \text{SWAP}_{S_B, R_B} \right)$ to the  state $\rho_{S_A  S_B} \otimes \left( \sigma_{1}^{(A)} \otimes \sigma_{1}^{(B)}\right)_{R_A  R_B}$. Once we can perform any two-subsystem unitary, we can use standard techniques to  construct a circuit to  implement any unitary on the system up to the desired accuracy\footnote{In particular, we can  apply the approach in \cite{nielsenandchuang}  (generalised to arbitrary dimension) based on two-level unitaries built from controlled unitary gates on two subsystems. We can then bound the total error by appropriately bounding the error on each gate.}.

If the subsystems are all of the same type, then an alternative approach would be to use {(i)} the fact that $U_{np} = \sqrt{\text{SWAP}}$ is an entangling two-subsystem gate which commutes with all conserved quantities, and {(ii)} that such a gate plus all single subsystem unitaries is computationally universal {\cite{DupontDupond}.} This would require only bipartite interactions, {with the caveat of} a more complicated construction. }

\section{An application to quantum thermodynamics}
\label{se:thermo}
Although the reference frame construction  above could be of use in 
many different scenarios, one interesting application is to  quantum 
thermodynamics with multiple  conserved quantities \cite{NCCQ,jonathan}. 
In this section, we briefly review the approach given in \cite{NCCQ}, 
and consider how our reference frame can be thought of as a battery in such cases. 

The general setup consists of a thermal bath, an additional system out of equilibrium with respect to the bath, and a battery that stores any work extracted from the system and bath. The thermal bath can be thought of as a collection of individual subsystems, each of which is in a generalised thermal state 
$\tau = \frac{1}{Z} e^{-(\beta_1 A_1 + \ldots + \beta_k A_k)}$, where  
$A_i$ are an arbitrary set of conserved quantities, $\beta_i$ are generalised 
inverse temperatures associated with these conserved quantities\footnote{E.g. for energy $\beta_i = \frac{1}{k_B T}$ where $k_B$ is Boltzmann's constant and $T$ is the  temperature, whereas for particle number conservation, $\beta_i = \frac{\mu_i}{k_B T}$ where $\mu_i$ is the chemical potential. Analogous inverse temperatures can be associated with other conserved quantities such as angular momentum.}
and $Z = \trace e^{-(\beta_1 A_1 + \ldots + \beta_k A_k)}$ is the generalised partition function. 
We will be particularly interested in cases in which the different conserved quantities do not commute. 

The simplest case to consider is a thermal bath with no additional system,  in which the battery is treated implicitly. In this scenario, one simply allows any unitary transformation on the thermal state of the bath, and assumes that any change in the conserved quantities of the  bath are accounted for by an equal and opposite amount of work of the corresponding type
\begin{align} 
W_{A_i}= -\Delta  \av{A_i^{\textrm{bath}}}, 
\end{align} 
in accordance with the first law of thermodynamics. A generalised second law of thermodynamics can then be derived stating that 
\begin{align}\label{eq:nosys}
\sum_i \beta_i W_{A_i} \leq 0\,,
\end{align}
where $W_{A_i}$ is the the type $A_i$ work extracted from the bath, and $\beta_i$ the inverse temperature conjugate to $A_i$.

In addition, if one now includes an additional non-thermal system in the picture, the amount of extracted work is bounded by
\begin{align}\label{eq:wsys}
\sum_i \beta_i W_{A_i} \leq -\Delta F_S\,,
\end{align}
where  $W_{A_i}= - \Delta \av{ A_i^{\textrm{sys}}} -\Delta  \av{A_i^{\textrm{bath}}}$ is the amount of $A_i$ type work extracted from the system and bath, and $F_S = \sum_i \beta_i \langle A_i^{\textrm{sys}} \rangle - S_S$ is the free entropy of the  system relative to the generalised bath, with $S_S = -\tr{\rho_S \ln \rho_S }{S}$ the system's entropy. 
 
Furthermore, given a small set of structural assumptions about the thermal bath, 
protocols can be constructed which can approximately implement any transformation 
satisfying the bounds given by eq.~\eqref{eq:nosys} and eq.~\eqref{eq:wsys}. 
Therefore, we learn that in principle we can extract as much $A_i$ type work $W_{A_i}$ 
as we want, as long as the other conserved quantities of the  bath and  system 
change as well  to compensate.

In the case in which the conserved quantities commute, and can 
in principle be stored separately from one another, one can extend these 
setups to include explicit quantum batteries \cite{NCCQ}. 
In particular, given a fixed initial state of a quantum  battery,  
we can map every  transformation in the implicit battery scenario onto 
a global unitary which commutes with all conserved quantities  
$A_i^{\textrm{sys}} + A_i^{\textrm{bath}} +A_i^{\textrm{battery}}$ 
and which achieves the same results to any desired accuracy. 

However, how to deal with explicit batteries in the presence of 
noncommmuting conserved quantities (under strict conservation) was left 
as an open question in \cite{NCCQ}, and treated using the $L^2(G)$
construction in \cite{jonathan}. Here we resolve this question by 
showing how the reference frames introduced earlier can be used as explicit 
batteries. Furthermore, because the same approach works for any 
conserved quantities, this also shows how explicit batteries can be 
constructed when the conserved quantities cannot be separated from each 
other, or when a battery system of the type in \cite{NCCQ} or \cite{jonathan} is not available. 

In particular, consider a bath and system  composed of a number of subsystems, and a reference frame composed of the same types of subsystems (as described in section  \ref{sec:composite}) which will act as a battery. Any desired protocol on the system and bath in the implicit battery scenario can be implemented approximately  to any desired accuracy given access to a sufficiently large reference frame, whilst respecting strict conservation of all extensive conserved quantities. Because of these conservation laws
\begin{align} 
\Delta \av{ A_i^{\textrm{sys}}}  + \Delta  \av{A_i^{\textrm{bath}}} + \Delta \av{ A_i^{\textrm{battery}}} =0.
\end{align} 
If the desired transformation on the system and bath would extract work $W_{A_i}$ in the implicit battery scenario, and this transformation is implemented with $\epsilon$ precision (in trace norm) using the reference frame, then 
\begin{align} 
| \Delta \av{ A_i^{\textrm{battery}}}  - W_{A_i} | \leq \epsilon \bigl\| A_i^{sys} + A_i^{bath} \bigr\|
\end{align} 
where $\| \cdot \|$ represents the operator norm. Any deviation between the work extracted in the implicit battery scenario, and the conserved quantities stored in the explicit battery can therefore be made as small as desired by taking a sufficiently large reference frame.

Following this procedure, any changes in the average values of  
conserved quantities in the battery correspond to stored or expended work.

\section{Conclusions}

{We have shown that it is possible to construct a simple reference frame, which allows us to apply any desired unitary on the system, and which moreover provides a physical intuition on the nature of the operations that should be applied to achieve it. This reference frame can in addition act as { a battery} of extensive conserved quantities.
{Previous general constructions of quantum reference frames \cite{RobRTA} have been based on the symmetry group associated with the conserved quantities via Noether's theorem. This is mathematically elegant, but requires systems and transformations which may be very difficult to find or implement physically. Our alternative } approach works universally for any set of conserved quantities, and does not  require Noether's theorem in the sense that we do not need to first construct the group symmetry belonging to the conservation law. {Furthermore, it allows for a relatively simple physical implementation involving a `circuit' model on multiple copies of the  system in a fixed product state.}}

The size of the frame in this ``bottom-up''
approach (compared to the ``top-down'' of $L^2(G)$), determines the 
accuracy with which we can implement a given unitary. Here we have 
prioritised the simplicity of the construction, rather than 
minimising its dimension or some other indicator of its complexity
for a given error. We leave it as an open problem whether more 
intricate constructions give a better accuracy for the same frame size.

\section*{Acknowledgements}
We thank Yakir Aharonov, Aram Harrow, Jonathan Oppenheim, Michalis Skontiniotis 
and Rob Spekkens for various discussions on reference frames over 
some time, which have helped us fix our own ideas.

This research was supported by Perimeter Institute for Theoretical Physics. 
Research at Perimeter Institute is supported by the Government of Canada 
through the Department of Innovation, Science and Economic Development Canada 
and by the Province of Ontario through the Ministry of Research, Innovation and Science. 
SP acknowledges support from the ERC Advanced Grant NLST and The Institute for Theoretical Studies, ETH Zurich.
AW acknowledges support by the Spanish MINECO (project FIS2016-86681-P) 
with the support of FEDER funds, and the Generalitat de Catalunya (project 2017-SGR-1127).


\appendix

\section{Partial trace relations involving $\text{SWAP}$} 
\label{ap:swap} 
Here we present two helpful relations involving the partial trace and the 
$\text{SWAP} = \sum_{i,j} \ket{ij}\!\bra{ji}$ operation.  
\begin{align} 
\tr{ \, A \otimes B \, \text{SWAP}}{B}  &= \sum_{ijk} \ket{i}_A\!\bra{ij} A \otimes B \, \text{SWAP}\ket{kj}\!\bra{k}_A \nonumber \\
&= \sum_{ijk} \ket{i}_A\!\bra{ij}  A \otimes B \ket{jk}\!\bra{k}_A \nonumber \\
& =  \sum_{ijk}  \proj{i} A \proj{j} B \proj{k} \nonumber \\
& = A B,
\end{align}
\begin{align} 
\tr{\text{SWAP} \, A \otimes B}{B}  &= \sum_{ijk} \ket{i}_A\!\bra{ij} \text{SWAP} A \otimes B \ket{kj}\!\bra{k}_A \nonumber \\
&= \sum_{ijk} \ket{i}_A\!\bra{ji}  A \otimes B \ket{kj}\!\bra{k}_A \nonumber \\
& =  \sum_{ijk}  \proj{i} B \proj{j} A \proj{k} \nonumber \\
& = B A.
\end{align}

\section{Bound on parameters $\alpha_k$} 
\label{ap:alpha_bound} 
Here we show that the parameters $\alpha_k$ appearing in $H= \sum_{k=1}^D \alpha_k \sigma_k$  (where $\| H\|\leq \pi$) are bounded by a constant. Consider the real vector space of Hermitian operators with Hibert-Schmidt inner product $\langle A,B \rangle= \trace{AB}$. Note that for any fixed operator basis $\{ \sigma_k \}_{k=0}^D$ (where for simplicity we have defined $\sigma_0 = \openone$), one can define a dual basis $\{ \tilde{\sigma}_l \}_{k=0}^D$ such that $\trace{\sigma_k \tilde{\sigma}_l } = \delta_{kl}$. It then follows that $\alpha_k = \trace{H \tilde{\sigma}_k}$ and hence from the Cauchy-Schwarz inequality that 
\begin{equation} 
|\alpha_k| \leq  \| H \|_2 \| \tilde{\sigma}_k \|_2 \leq \sqrt{D} \pi \| \tilde{\sigma}_k \|_2 
\end{equation} 
where  $\| H \|_2  = \sqrt{\trace{H^2}}$ is the Hilbert-Schmidt norm.  Hence 
\begin{equation} 
\alpha_{\max} = \sqrt{D} \pi \max_k \| \tilde{\sigma}_k \|_2  
\end{equation} 
provides an upper bound for all parameters $\alpha_k$.

\section{Overall error bound due to product of small rotations} 
\label{concatenation} 
Here we prove that performing a sequence of $N$ transformations involving the reference frame, each of which approximately implements $U_H = \exp\left( - i \frac{H}{N} \right)$ up to error $\cO \left( \frac{1}{N^2} \right)$, will implement the overall transformation $U_S =(U_H)^N$ up to error  $\cO \left( \frac{1}{N} \right)$. Let us assume that after $t$ steps, the system is in a state $\rho_S^t$ which satisfies:
\begin{align}
  \left\| \rho_S^t - (U_H)^t \, \rho_S  \, (U_H^{\dagger})^t \right\|_1 \leq\delta, 
\end{align}
where $\|\cdot\|_1$ denotes the trace norm. 
Hence, 
\begin{align*}
&\left\| \tr{V_{\text{seq}} \rho_S^t \otimes \sigma_R \,V_{\text{seq}}^\dagger}{R_1 \ldots R_D} - (U_H)^{t+1} \, \rho_S  \, (U^{\dagger}_H)^{t+1} \right\|_1 \\
&\phantom{=====}
 \leq \left\| U_H \rho_S^t U_H  + \cO\left( \frac{1}{N^2} \right) - (U_H)^{t+1} \, \rho_S  \, (U^{\dagger}_H)^{t+1} \right\|_1 \\
&\phantom{=====}
 \leq \left\| U_H \left( \rho_S^t - (U_H)^t \, \rho_S  \, (U^{\dagger}_H)^t \right) U^{\dagger}_H \right\|_1  + \cO \left( \frac{1}{N^2} \right)  
 \leq \delta + \cO \left( \frac{1}{N^2} \right) . 
\end{align*}
Since at $t=1$ eq.~\eqref{eq:nueqvoe} implies $\delta=\cO \left( \frac{1}{N^2} \right) $, after the sequence of $N$ applications of $V_{\text{seq}}$ the error bound is $\cO \left( \frac{1}{N} \right)$, and the unitary effectively applied on the system is $(U_H)^N = U_S$.

\section{Further details on the general $\cO\left( \frac{1}{N^2} \right)$ and $\cO\left( \frac{1}{N} \right)$ bounds}
\label{ap:On}

Here we provide more detailed technical proof of our  $\cO\left( \frac{1}{N^2} \right)$ and $\cO\left( \frac{1}{N} \right)$ bounds, showing  that they are indeed bounded in the required sense, and giving an explicit bound in each case showing their dependence on the various constants. In the first case, we wish to show that 
\begin{align} 
\left\| \tr{ V_{\alpha} \, \rho_S \otimes \sigma \, V_{\alpha}^{\dagger} }{R} - (\rho_S - i \frac{\alpha}{N}  [ \sigma, \rho_S ]) \right\|_1\leq \cO\left( \frac{1}{N^2} \right).
\end{align} 
where $V_{\alpha} =\exp \left( -i \frac{\alpha}{N} \text{SWAP} \right)$. Expanding  $V_{\alpha}$ and $V_{\alpha}^{\dagger}$  as  power series in $\frac{1}{N}$ and using \eqref{eq:effectiveU} we obtain  
\begin{align} \label{eq:singleclose} 
\left\|  \tr{ V_{\alpha} \, \rho_S \otimes \sigma \, V_{\alpha}^{\dagger} }{R} - (\rho_S - i \frac{\alpha}{N}  [ \sigma, \rho_S ]) \right\|_1  &= \left\| \sum_{n=2}^\infty  \sum_{k=0}^n \tr{ \frac{(-i \frac{\alpha}{N}\, \text{SWAP})^k}{k!} \rho_S \otimes \sigma  \frac{(i \frac{\alpha}{N} \, \text{SWAP})^{n-k}}{(n-k)!}}{R} \right\|_1 \nonumber \\
& \leq  \sum_{n=2}^\infty \left(  \frac{\alpha}{N}\right)^n \sum_{k=0}^n \frac{1}{k! (n-k)!} \, \left\| \text{SWAP}^k  \rho_S \otimes \sigma  \, \text{SWAP}^{n-k}  \right\|_1 \nonumber \\
& \leq  \sum_{n=2}^\infty \left(  \frac{\alpha}{N}\right)^n \sum_{k=0}^n \frac{1}{k! (n-k)!} \nonumber \\
& = \sum_{n=2}^{\infty} \frac{1}{n!} \left(  \frac{\alpha}{N}\right)^n \sum_{k=0}^n \frac{n!}{k! (n-k)!} \nonumber \\
& =  \sum_{n=2}^{\infty} \frac{1}{n!} \left(  \frac{\alpha}{N}\right)^n (2^n) \nonumber \\
& =  \sum_{n=2}^{\infty} \frac{1}{n!} \left(  \frac{2\alpha}{N}\right)^n  \nonumber \\
& = 4 \left(  \frac{\alpha}{N}\right)^2  \sum_{n=2}^{\infty} \frac{1}{n!} \left(  \frac{2\alpha}{N}\right)^{n-2}. 
\end{align} 
where we have used the fact that the trace norm satisfies $\| \tr{O_{AB}}{B} \|_1 \leq \| O_{AB} \|_1$  \cite{rastegin} and  $\| A\|_1 = \| U A V\|_1$  for unitary $U$ and $V$. For all $N \geq 2 \alpha$ we therefore have 
\begin{align} \label{eq:singleclose2} 
\left\|  \tr{ V_{\alpha} \, \rho_S \otimes \sigma \, V_{\alpha}^{\dagger} }{R} - (\rho_S - i \frac{\alpha}{N}  [ \sigma, \rho_S ]) \right\|_1 & \leq  4 \left(  \frac{\alpha}{N}\right)^2  \sum_{n=2}^{\infty} \frac{1}{n!} \nonumber \\
& =  4 (e-2) \left(  \frac{\alpha}{N}\right)^2 .
\end{align} 
Hence $\left\|  \tr{ V_{\alpha} \, \rho_S \otimes \sigma \, V_{\alpha}^{\dagger} }{R} - (\rho_S - i \frac{\alpha}{N}  [ \sigma, \rho_S ]) \right\|_1 = \cO\left( \frac{1}{N^2} \right)$. 
Note that this is not the tightest provable bound, and a higher threshold for $N$ can generate a  smaller (yet more complicated)  constant, but this is not the focus of this paper.

For the next stage of the proof, we must show that
\begin{align} 
\left\|  \exp\left( -i\frac{\alpha}{N} \sigma \right) \,\rho_S\, \exp \left( +i\frac{\alpha}{N} \sigma \right) - (\rho_S - i \frac{\alpha}{N}  [ \sigma, \rho_S ]) \right\|_1\leq \cO\left( \frac{1}{N^2} \right). 
\end{align} 
 Expanding the exponentials and proceeding almost identically to \eqref{eq:singleclose} and \eqref{eq:singleclose2}, for all $N \geq 2 \alpha$ we obtain 
\begin{align} \label{eq:singleclose3} 
\left\|  \exp\left( -i\frac{\alpha}{N} \sigma \right) \,\rho_S\, \exp \left( +i\frac{\alpha}{N} \sigma \right) - (\rho_S - i \frac{\alpha}{N}  [ \sigma, \rho_S ]) \right\|_1  &= \left\| \sum_{n=2}^\infty  \left(  \frac{\alpha}{N}\right)^n \sum_{k=0}^n \frac{(-i \sigma)^k}{k!} \rho_S \frac{(i \sigma)^{n-k}}{(n-k)!}  \right\|_1 \nonumber \\
& \leq  \sum_{n=2}^\infty \left(  \frac{\alpha}{N}\right)^n \sum_{k=0}^n \frac{1}{k! (n-k)!} \nonumber \\
& \leq  4 (e-2) \left(  \frac{\alpha}{N}\right)^2 .
\end{align} 
Hence, $\left\|  \exp\left( -i\frac{\alpha}{N} \sigma \right) \,\rho_S\, \exp \left( +i\frac{\alpha}{N} \sigma \right) - (\rho_S - i \frac{\alpha}{N}  [ \sigma, \rho_S ]) \right\|_1= \cO\left( \frac{1}{N^2} \right)$. Combining these two results using the triangle inequality, we obtain 
\begin{align}
\|\tr{V_{\alpha} \, \rho_S \otimes \sigma\, V_{\alpha}^\dagger }{B} - U_{\sigma}(\alpha) \, \rho_S  \, U_{\sigma}^{\dagger}(\alpha) \|_1 \leq  8(e-2) \left(  \frac{\alpha}{N}\right)^2 
\end{align}
for all  $N \geq 2 \alpha$. This is an explicit form of the bound given in \eqref{eq:epsfor1}.

We now consider the bounds for a general small rotation. First consider 
\begin{align} \label{eq:general_rot}
\eta = \left\| \tr{V_{\alpha_D}^{(D)} \ldots V_{\alpha_2}^{(2)} V_{\alpha_1}^{(1)} ( \rho_S \otimes \sigma_R) \, V_{\alpha_1}^{(1)\dagger}  V_{\alpha_2}^{(2)\dagger} \ldots  V_{\alpha_D}^{(D)\dagger}}{R_1 \ldots R_D} - ( \rho_S -i \tfrac{1}{N} \, [H, \rho_S] ) \right\|_1 .
\end{align} 
Expanding each of the unitaries as a power series in $\frac{1}{N}$ and collecting terms in $\left( \frac{1}{N}\right) ^n$ we note that the terms for $n=0$ and $n=1$ cancel (as in the proof in the main paper). The number of terms of order $n$ is equal to the number of ways of distributing $n$ indistinguishable balls in $(2D)$ distinguishable boxes (as we can choose which of the $2D$ unitaries $V^{(k)}_{\alpha_k}$ or $V_{\alpha_k}^{(k)\dagger}$   generates each power of $\left( \frac{1}{N}\right)$). For simplicity, we upper bound this by the number of ways of distributing $n$ distinguishable balls in $(2D)$ distinguishable boxes, which is $(2D)^n$. Each term also has a constant coefficient whose magnitude  is upper bounded by $\alpha_{\max}^n$ where  $\alpha_{\max} \geq |\alpha_k|$ for all $k$ and $U_S$. A derivation of the  constant   $\alpha_{\max}$ is given in appendix \ref{ap:alpha_bound}. 

Finally each term is associated with an operator which is the trace over the reference frame of some sequence of SWAP operations before and after $\rho_S \otimes \sigma_R$. Writing such a term as $\tr{\Pi_1 \, \rho_S \otimes \sigma_R \Pi_2}{R}$  where $\Pi_1$ and $\Pi_2$ are permutations, we see that 
\begin{equation} 
\| \tr{\Pi_1 \, \rho_S \otimes \sigma_R \Pi_2}{R} \|_1 \leq \| \Pi_1 \, \rho_S \otimes \sigma_R \Pi_2  \|_1 =  \| \rho_S \otimes \sigma_R \|_1 = 1
\end{equation} 
Combining all of these observations, we find that 
\begin{align}
\eta \leq  \sum_{n=2}^{\infty} \left( \frac{2 D \alpha_{\max}}{N}\right)^n.
\end{align}  
As long as $N \geq 4D \alpha_{\max}$, this gives 
\begin{equation}
\eta \leq \left( \frac{2 D \alpha_{\max}}{N}\right)^2 \left( 1 + \frac{1}{2} + \frac{1}{4}  + \ldots \right)  
     =    8 \left( \frac{ D \alpha_{\max}}{N}\right)^2.
\end{equation}  
Hence $\eta \leq \cO\left( \frac{1}{N^2} \right)$. Following an almost identical approach to \eqref{eq:singleclose}, \eqref{eq:singleclose2} and \eqref{eq:singleclose3}, one can also show that 
\begin{align} \label{eq:closeH}
\left\|  \exp\left( -i\frac{H}{N}  \right) \,\rho_S\, \exp \left( +i\frac{H}{N}  \right) - (\rho_S - i \frac{1}{N}  [ H, \rho_S ]) \right\|_1 & \leq 4(e-2) \left( \frac{\| H \|_1}{N}\right)^2  \nonumber \\
& \leq 4 (e-2) \left( \frac{\pi \sqrt{D+1} }{N}\right)^2,
\end{align} 
where in the last line, we have used the fact that the dimension of the system space is $\sqrt{D+1}$, and that we can choose $H$ such that it's maximum eigenvalue is $\pi$. Combining this result with the result for $\eta$ we obtain 
\begin{align} 
\left\| \tr{V_{\text{seq}}  \rho_S \otimes \sigma_R \, V_{\text{seq}}^\dagger}{R_1 \ldots R_D} - U_H \rho_S U^{\dagger}_H  \right\|_1 \leq \left(8 D^2 \alpha_{\max}^2 + 4 \pi^2 (e-2) (D+1) \right) \frac{1}{N^2} 
\end{align} 
for sufficiently large $N$.  Thus using the inductive argument in appendix \ref{concatenation} gives 
\begin{equation} 
\left\| \tr{V \, \rho_S \otimes \rho_R\, V^\dagger }{R} - U_S \rho_S U_S^\dagger \right\|_1 \leq \left(8 D^2 \alpha_{\max}^2 + (2 \pi)^2 (e-2) (D+1) \right) \frac{1}{N}. 
\end{equation} 

\end{document}